\documentclass{article}

\usepackage{amsmath}
\usepackage{amssymb}
\usepackage{amsthm}
\usepackage{algorithm}
\usepackage{algorithmic}
\usepackage{graphicx}
\usepackage{setspace}
\usepackage{bussproofs}
\usepackage{mathdots}

\EnableBpAbbreviations
\newtheorem{thm}{Theorem}[section]

\newtheorem{lem}[thm]{Lemma}

\theoremstyle{definition}
\newtheorem{defn}[thm]{Definition}
\newtheorem{notation}{Notation}
\theoremstyle{remark}

\numberwithin{equation}{section}
\newcommand{\set}[1]{\left\{#1\right\}}

\newcommand{\lfl}{\mathcal L _{FL}}
\newcommand{\la}{\mathcal L ^2_A}

\newcommand{\AND}{\wedge}
\newcommand{\false}{\bot}
\newcommand{\true}{\top}
\newcommand{\OR}{\vee}

\newcommand{\vl}{\overline{VL}}
\renewcommand{\implies}{\supset}
\newcommand{\size}[1]{|#1|}
\newcommand{\pair}[1]{\langle #1 \rangle}
\newcommand{\same}{\equiv}
\renewcommand{\iff}{\leftrightarrow}

\newcommand{\proves}{\vdash}

\newcommand{\sob}{\Sigma_0^B}
\newcommand{\sqcnf}{\Sigma CNF(2)}
\newcommand{\fCenter}{\longrightarrow}
\newcommand{\satisfies}{\models}

\renewcommand{\implies}{\supset}

\newcommand{\intersect}{\cap}

\begin{document}

\title{Quantified Propositional Logspace
Reasoning}%
\author{Steven Perron}%
% ----------------------------------------------------------------
\maketitle

\abstract{In this paper, we develop a quantified propositional proof systems that corresponds to logarithmic-space reasoning. 
We begin by defining a class $\sqcnf$ of
quantified formulas that can be evaluated in
log space.  Then our new proof system $GL^*$ is defined as $G_1^*$ with cuts
restricted to $\sqcnf$ formulas and no cut formula
that is not quantifier free contains a free
variable that does not appear in the final formula.

To show that $GL^*$ is strong enough to capture log space
reasoning, we translate theorems of $VL$ into a family of
tautologies that have polynomial-size $GL^*$ proofs.
$VL$ is a theory of bounded arithmetic that is known
to correspond to logarithmic-space reasoning.  To do
the translation, we find an
appropriate axiomatization of $VL$, and put $VL$
proofs into a new normal form.

To show that $GL^*$ is not too strong, we prove the
soundness of $GL^*$ in such a way that it can be formalized
in $VL$.  This is done by giving a logarithmic-space algorithm that
witnesses $GL^*$ proofs.
}

\section{Introduction}

Recently there has been a significant amount of research looking into the
connection between computational complexity, bounded
arithmetic, and propositional proof complexity.  A recent
survey on this topic can be found at \cite{CN06}.  The idea
is that there is a hierarchy of complexity classes
$$AC^0 \subsetneq TC^0 \subseteq NC^1 \subseteq L \subseteq
NL \subseteq P.$$
The first class is the set of problems that can be solved by
uniform, polynomial-size, constant depth circuits.  This
class is important because it can be shown that PARITY
cannot be solved in $AC^0$.  In fact, problems that
involve counting cannot be solved in $AC^0$.  The second
class is $TC^0$.  This set of problems is the same as $AC^0$
except that $TC^0$ circuits can use counting gates.  The
class $NC^1$ is the set of problems that can be solved using
polynomial-size, logarithmic-depth circuits.  This class can
be thought of as the set of problems that can be solved very
quickly when work is done in parallel.  Evaluating
boolean formulas is complete for this class.  The class $L$
is the set of problems that can be solved in
logarithmic space on a Turing machine.  The class $NL$ is
the set of problems that can be solved in logarithmic space
on a non-deterministic Turing machine.  The reachability
problem for directed graphs is complete for this class.  The
sequence finishes with $P$, the set of problems that can be
solved in polynomial time on a deterministic Turing machine.
Except for the first inclusion, is it unknown if any of these
inclusions are proper.  

Each of these complexity classes has
a corresponding theory of arithmetic:
$V^0, VTC^0, VNC^1, VL, VNL,$ and $TV^0$, respectively.
Each of these theories can prove that the functions in their
corresponding complexity class are total.  As a consequence,
any information we can obtain about the theory tells us
something about the complexity class and vice versa.

There is also a connection with propositional proof
complexity.  Some of the theories mentioned above have a
corresponding propositional proof system.  As before,
information about the proof systems tells us about the
corresponding theory and complexity class.  In this paper,
we explore the proof systems.
The goal is to try to understand how the strength
of a proof system is affected by different restrictions.

Our focus will be on quantified propositional proof systems,
but, to explain our method, we will use quantifier-free
propositional proof systems.
Start with a Frege proof
system, sometimes called Hilbert Style Systems.  These systems are described in standard logic text
books.  A Frege proof is a series of propositional formulas
where each formula is an axiom or can be inferred from
previous formulas using one of the rules of
inference.
There are two common ways of restricting this proof system.
The first is to restrict all of the formulas in the proof.
For example, one definition of bounded-depth Frege is to
restrict every formula in the proof to formulas with a
constant depth.  This worked, but, if a proof system is
defined this way, then there are formulas that cannot be
proved simply because they are not allowed to appear in the
proof.  For example, bounded-depth Frege with formulas of
depth $d$ cannot prove any formula of depth $d+1$.
The other method is to restrict the formulas on which certain
rules can be applied.  This solves the problem of the first
method and led to other definitions of bounded-depth Frege.

In this paper, we will look at restricting the cut rule in
the tree-like sequent calculus for quantified propositional formulas.
This systems is known as $G^*$.  The cut rule derives
$\Gamma \fCenter \Delta$ from $A, \Gamma \fCenter \Delta$
and $\Gamma \fCenter \Delta, A$.  In $G^*$, $A$ can be any quantified propositional formulas.
The proof system $G_0^*$ is defined by restricting $A$ to quantifier-free formulas.  If we are
given a $G_0^*$ proof of a $\Sigma_1^q$ formula ($\exists \vec z B(\vec
z)$, where $B$ is quantifier-free), then we can find a
witness for existential quantifiers in this formula in uniform $NC^1$; moreover, this
problem is complete for this class.  The complexity
class $NC^1$ is the set of problems that can be solved by
polynomial-size, logarithmic-depth circuits with fan-in 2.
The interesting observation is that evaluating
quantifier-free formulas is also complete for $NC^1$.
It is also possible to connect $G_0^*$ to $NC^1$ indirectly
through bounded arithmetic.  There is a theory of arithmetic
$VNC^1$ that is known to correspond to $NC^1$ reasoning.
Given a $VNC^1$ proof of a bounded formula it is possible to translate the proof
into a family of polynomial-size $G_0^*$ proofs.  This tells
us that the reasoning power of $G_0^*$ is at least as strong
as that of $VNC^1$ \cite{CM04}.  In the other direction, $VNC^1$ can
prove that $G_0^*$ is sound when proving $\Sigma_1^q$
formulas.  This means that, when proving
$\Sigma_1^q$ formulas, the reasoning power of $G_0^*$ is
not stronger than that of $VNC^1$.  So we say that $G_0^*$
corresponds to $NC^1$ reasoning.

As well, if we restrict cut formulas to
constant-depth, quantifier-free formulas, we get a proof system
that corresponds to $AC^0$ reasoning.  The complexity class
$AC^0$ is the set of problems that can be solved by
polynomial-size, constant-depth circuits with unbounded
fan-in.  Again, evaluating constant-depth
formulas is complete for $AC^0$.  We should note we are
talking about the proofs of quantifier-free formulas.

This gives us two proof systems whose reasoning power
is the same as the complexity of evaluating their cut
formulas.  This raises the
question of whether or not this holds in general.
The quick answer is no.  A counter-example is $G_1^*$.
Evaluating $\Sigma_1^q$ formulas is complete for $NP$, but
the $\Sigma_1^q$ witnessing problem for $G_1^*$ is complete
for $P$ \cite{Krajicek95}.  Another counter-example is $GPV^*$, where cut
formulas are quantifier-free or formulas of the form $\exists x
[ x \iff A]$, where $A$ is a quantifier-free formula that does
not mention $x$.  Evaluating a cut formula for $GPV^*$ is
complete for $NC^1$, but the witnessing problem is complete
for $P$ \cite{Perron07}.

In this paper, we define a new proof system $GL^*$ that corresponds
to $L$ reasoning.  The complexity class $L$ is the set of
problems that can be solved on a Turing Machine with a
read-only input tape and a work tape where the space used on
the work tape is proportional to the logarithm of the size
of the input.
Our proof system $GL^*$ is defined by restricting cuts to
$\sqcnf$ formulas, a set of formulas for which the
evaluation problem is complete for $L$.  However, that is
not enough.  We also restrict the free variables that appear
in cut formulas with quantifiers to variables that appear free
in the final sequent.  We then prove this proof system
corresponds to $L$ reasoning by connecting it with a theory
of arithmetic that is known to correspond to $L$ reasoning.
  This definition is meant to demonstrate that
the strength of a proof system is not related to the
difficulty of evaluating a single cut formula in the proof,
but to the complexity of witnessing the eigenvariables in the proof.

In Section 2, we give definitions of the
important concepts.  In particular, we define
two-sorted computational complexity and bounded arithmetic.
As well, we define the standard proof systems and explain
the connection between proof systems and theories of bounded
arithmetic in more detail.  In Section 3, we define $GL^*$.
This includes the definition of the $\sqcnf$ formulas.  In
Section 4, we change the theory $VL$ and prove a normal-form
that is necessary for our results.  This is the most
technical section in the paper.  In Section 5, we prove
the translation theorem.  In Section 6, we prove that $GL^*$
is sound in the theory.  This includes an algorithm to
evaluate $\sqcnf$ formulas in $L$.

This paper is an expanded version of the author's earlier paper \cite{Perron05}.

\section{Basic Definitions And Notation}

\subsection{Two-Sorted Computational Complexity}

In this paper, we use two-sorted computational complexity.
The two sorts are numbers and binary strings (aka
finite sets).  The numbers are intended to range over the
natural numbers and will be denoted by lower-case letters.
For example, $i$, $j$, $x$, $y$, and $z$ will often be used
for number variables; $r$, $s$, and $t$ will be used for
number terms; and $f$, $g$ and $h$ will be used for
functions that return numbers.  The strings are intended to be
finite strings over $\set{0,1}$ with leading $0$ removed. Since the strings are finite,
they can be thought of as sets where the $i$th bit is 1
if $i$ is in the set.  The strings will
be denoted by upper- case letters.  The letters $X$,$Y$, and
$Z$ will often be used for string variables.

We focus on the complexity class $L$.  Let $R(\vec x, \vec
X)$ be a relation.  If we are going to solve this relation
on a Turing Machine $M$, then the input to $M$ will be
$\vec x$ in unary and $\vec X$ as a series of binary
strings.  So the size of the input is $\vec x + \size {\vec
X}$.  We say $R$ is in $L$ if $R$ can be decided by a
two-tape Turing Machine such that one tape is a read-only
input tape, and less than $O(\log(\vec
x + \size {\vec X}))$ squares are visited on the other tape.

For functions, we say a number function $f(\vec x, \vec
X)$ is in $FL$ if there is a polynomial $p$ such that $f(\vec x, \vec X) < p(\vec x, \size {\vec
X})$, and the relation $f(\vec x, \vec X) = y$ is in $L$.  A
string function $F(\vec x, \vec X)$ is in $FL$ if the size
of $F(\vec x,\vec X)$ is bounded by a polynomial and if
the relation 
$$R(i, \vec x, \vec X) \iff \text{ the $i$th bit of $F(\vec x, \vec X)$ is 1}$$ is in $L$.
This is equivalent to defining $FL$ using a three-tape Turing
Machine with a write-only output tape.

\subsection{Two-Sorted Bounded Arithmetic}
\label{sec:def}

Besides two-sorted computational complexity, we also use the two-sorted bounded arithmetic.
The sorts are the same.  This notation was base on the work
of Zambella in \cite{Zambella96}, but 
we follow the presentation of Cook and Nguyen from
\cite{Cook05,CN06}.

The base language is $$\la = \set{ 0,1,+, \times, <, =, =_2, \in,
\size{} }.$$  The constants $0$ and $1$ are number constants.
The functions $+$ and $\times$ take two numbers as input and return a
number--the intended meanings are the obvious ones.  The language
also includes two binary predicates that take two numbers:
$<$ and $=$.  The predicate $=_2$ is meant to be equality
between strings, instead of numbers.  In practice, the $2$
will not be written because which equality is meant is
obvious from the context.  The membership predicate $\in$ takes
a number $i$ and a string $X$.  It is meant to be true if
the $i$th bit of $X$ is 1 (or $i$ is in the set $X$). This
will also be written as $X(i)$.  The final function $\size X$
takes a string as input and returns a number.  It is intended to be
the number of bits needed to write $X$ when leading zeros
are removed (or the least upper bound  of the set $X$).
The set of axioms 2BASIC is the set of defining axioms for $\la$.
\smallskip
\begin{tabular}{ll}
$B1$. $x+1 \neq 0$                      & $B7$. $x \le x+y$ \\
$B2$. $x+1 = y+1 \implies x = y$        & $B8$. $(x \le y \AND y \le x) \implies x = y$ \\
$B3$. $x+0 = x$                         & $B9$. $0 \le x$ \\
$B4$. $x+(y+1)=(x+y)+1$                 & $B10$. $x \le y \OR y \le x$ \\
$B5$. $x \times 0 = 0$                  & $B11$. $x \le y \iff x < y+1$ \\
$B6$. $x \times (y+1) = (x \times y)+x$ & $B12$. $x \neq 0 \implies \exists y \le x (y+1 = x)$ \\
$L1$. $X(y) \implies y < \size{X}$      & $L2$. $y+1=\size{X} \implies X(y)$ \\
\multicolumn{2}{l}{$SE$. $X=Y \iff [ \size{X} = \size{Y} \AND \forall i < \size{X}(X(i) \iff Y(i))]$} \\
\end{tabular}
\smallskip

We use $\exists X < b ~ \phi$ as shorthand for $\exists X[
(\size X < b) ~ \AND ~ \phi]$.  The shorthand $\forall X < b ~ \phi$
means $\forall X [ (\size X < b) \implies \phi ]$.  The set
$\Sigma_0^B = \Pi_0^B$ is the set of formulas whose only quantifiers
are bounded number quantifiers.  For $i > 0$, the set
$\Sigma_i^B$ is the set of formulas of the form
$\exists \vec X < \vec t \phi$ where $\phi$ is a $\Pi_{i-1}^B$
formula.
For $i > 0$, the set
$\Pi_i^B$ is the set of formulas of the form
$\forall \vec X < \vec t \phi$ where $\phi$ is a $\Sigma_{i-1}^B$
formula.

Now we can define two important axiom schemes:
\begin{align*}
\Phi\text{-COMP:}~&
\exists X \le b \forall i < b [ X(i) \iff \phi(i) ], \\
\Phi\text{-IND:}~&
[\phi(0) \AND \forall x < b [ \phi(x) \implies \phi(x+1)
] ] \implies \phi(b)
\end{align*}
where $\Phi$ is a set of formula and $\phi(i) \in \Phi$, and, for
$\Sigma_i^B$-COMP,
$\phi$ does not contain $X$, but may contain other free variables.

We can now define the base theory.
\begin{defn}
The theory $V^0$ is axiomatized by the 2BASIC
axioms plus $\Sigma_0^B$-COMP.
\end{defn}
It is possible to show that $V^0$ proves
$\Sigma_0^B$-IND (Corollary \cite{CN06}).
This theory is typically viewed at the theory that corresponds to $AC^0$ reasoning.  

From time to time, we will use functions symbols that are
not in $\la$.  The first is $X(i,j) \equiv X(\pair{i,j})$,
where $\pair{i,j} = (i+j)(i+j+1)+2j$ is the pairing function.
It can be thought of as a two dimensional array of bits.  The
second is the row function.  The notation we use is
$X^{[i]}$.  This functions returns the $i$th
row of the two dimensional array $X$.  In the same way, we
can also describe three dimensional arrays.  We also want to
pair string.  So if $X = \pair{Y_1, Y_2}$, then $X^{[0]} =
Y_1$ and $X^{[1]} = Y_2$.  Note that, if we add these
functions with their $\Sigma_0^B$ defining axioms to any
theory $T$ extending $V^0$, we get a conservative extension.  They can
also be used in the induction axioms \cite{Cook05}.  This
means that, if there is a $T$ proof of a formula that uses
these functions, there is a $T$ proof of the same formula
that does not use these functions.

To get a theory that corresponds to $L$ reasoning, we add an
axiom that says there is an output to a function that is
complete for $L$ with respect to $AC^0$ reductions.  This is
a specific example of
the method used in \cite{Cook05} to construct a theory for a
given complexity class.  The theory we define is
$\Sigma_0^B$-rec from \cite{Zambella97}, but we will call it
$VL$.  The complete
function we use is: Given a graph with edge relation
$\phi(i,j)$ and nodes $\set{0,\dots, a}$, where
every vertex in the
graph has out-degree at least
$1$, find a path of length $b$.
This is expressed using the $\Sigma_0^B$-rec axiom:
\begin{equation}
\tag{$\Sigma_0^B$-rec}
\label{xrec}
\begin{split}
\forall x \le a \exists y \le a \phi(x,y) & \implies
\exists Z, \forall w \le b \phi( f(a,w,Z), f(a,w+1,Z) )
\end{split}
\end{equation}
where $f(a,w,Z) = \min \limits_x ~ (Z(w,x) \OR x = a)$ and
$\phi$ is a $\Sigma_0^B$ formula.  The idea is that the
function $f(a,w,Z)$ extracts the $w$th node in the path that
$Z$ encodes.
\begin{defn}
The theory $VL$ is the theory axiomatized by $V^0$ plus $\Sigma_0^B$-rec.
\end{defn}
The $\Sigma_0^B$-rec axiom has the disadvantage that the
path can start at any node.  However, as
Zambella pointed out in \cite{Zambella97}, to is possible to prove that there is a
path of length $b$ starting at a particular node $a$.
\begin{lem}
\label{lem:start}
Let $E$ be the edge relation for a directed graph on the
nodes $\set{0, \dots, n-1}$.  Then for all $a < n$ and $b$, $VL$
proves, if $\forall i < n \exists j
< n~E(i,j)$, then there is a path of length $b$ starting at node $a$.
\end{lem}
\begin{proof}
Define $\phi(\pair{w,i},\pair{w',j})$ as
$$\phi(\pair{w,i},\pair{w',j}) \equiv (w' = w+1 \mod b+1) \AND (w' \neq 0
\implies E(i,j)) \AND (w' = 0 \implies j = a).$$
Take a path of length $2b$ in the graph of
$\phi$.  At some point in the first half of that path, the
path passes through the node $\pair{0,a}$.  Starting from
there we can extract a path of length $b$ in $E$ that starts
at node $a$.
\end{proof}

\subsection{A Universal Theory For $L$ Reasoning}

Another way to get a theory for $L$ is to define a universal
theory with a language that contains a function symbol for
every function in $FL$.
Then, we get a
theory for $L$ by taking the defining axioms for these
functions.  This is the idea behind other universal theories
like $PV$ and
$\overline{V^0}$.
In our case, we characterize the $FL$ functions using Lind's
characterization \cite{Lind74} adjusted for the two-sort
setting.

In the next
definition, we define the set of function symbols in $\lfl$
and give their intended meaning.

\begin{defn}
\label{def_lfl}
The language $\lfl$ is the smallest language satisfying

\begin{enumerate}
\item $\la \cup \set{pd, \min }$ is a subset of $\lfl$ and have
defining axioms 2BASIC, and the
axioms
\begin{eqnarray}
\label{pdaxiom1}
pd(0)=0\\
\label{pdaxiom2}
pd(x+1)=x \\
\label{minaxiom}
min(x,y)=z \iff (z=x \AND x \le y) \OR (z=y \AND y \le x)
\end{eqnarray}

\item For every open formula $\alpha(i,\vec x, \vec X)$ over $\lfl$
and term $t(\vec x, \vec X)$ over $\la$, there is a string function
$F_{\alpha,t}$ in $\lfl$ with bit defining axiom

\begin{equation}
\label{lflString}
 F_{\alpha,t}(\vec x, \vec X)(i) \iff i < t(\vec x, \vec
X) \AND \alpha(i, \vec x, \vec X)
\end{equation}

\item For every open formula $\alpha(z,\vec x, \vec X)$ over $\lfl$
and term $t(\vec x, \vec X)$ over $\la$, there is a number function
$f_{\alpha,t}$ in $\lfl$ with defining axioms

\begin{eqnarray}
\label{lflNumber1}
f_{\alpha,t}(\vec x, \vec X) \le t( \vec x, \vec X) \\
\label{lflNumber2}
z < t(\vec x, \vec X) \AND \alpha(z, \vec x,
\vec X)
\implies \alpha(f_{\alpha,t}(\vec x, \vec X), \vec x, \vec X) \\
\label{lflNumber3} z < f_{\alpha,t}(\vec x, \vec X) \implies \neg
\alpha(z, \vec x, \vec X)
\end{eqnarray}

\item For all number functions $g(\vec x, \vec X)$ and $h(p, y, \vec x, \vec X)$
in $\lfl$ and term $t(y, \vec x, \vec X)$ over $\la$, there is a number function
$f_{g,h,t}(y, \vec x, \vec X)$ with defining axioms

\begin{eqnarray}
\label{pnumrec1}f_{g,h,t}( 0, \vec x, \vec X) = \min(g(\vec x, \vec X),t(\vec x, \vec X)) \\
\label{pnumrec2}f_{g,h,t}( y+1, \vec x, \vec X ) = \min (h(
f(y, \vec x, \vec X, y, \vec x, \vec X) ), t(\vec x, \vec X)) 
\end{eqnarray}
\end{enumerate}
\end{defn}

The last scheme is called $p$-bounded number recursion.
The $p$-bounded number recursion is equivalent to
the $log$-bounded string recursion given in \cite{Lind74}.
The other schemes come from the definition of
$\mathcal{L}_{FAC^0}$ in \cite{Cook05}.

It is not difficult to see every function in $\lfl$ is in
$FL$.  The only
point we should note is that the intermediate values in the
recursion are bounded by a polynomial in the size of the input.
This means, if we store intermediate values in binary, 
the space used is bounded by the log of the size of the input.  So the
recursion can be simulated in $L$.  To show that every
$FL$ function has a corresponding function symbol in $\lfl$,
note that the $p$-bounded number recursion can be used to
traverse a graph where every node has out-degree at most
one.

\begin{defn}
$\vl$ is the theory over the language $\lfl$ with
B1-B11, SE, plus \ref{pdaxiom1};
\ref{pdaxiom2}; \ref{minaxiom};
axiom \ref{lflString} for each string function
$F_{\alpha,t}$ in $\lfl$; axioms \ref{lflNumber1}, \ref{lflNumber2},
and \ref{lflNumber3} for each number function 
$f_{\alpha,t}$ in $\lfl$; and axioms \ref{pnumrec1} and
\ref{pnumrec2} for each number function $f_{g,h,t}$ in
$\lfl$.
\end{defn}
An open($\mathcal L$) formula is a formula over the language
$\mathcal L$ that does not have any
quantifiers.

The important part of this theory is that it really is a universal
version of $VL$.
\begin{thm}
$\vl$ is a conservative extension of $VL$.
\end{thm}
\begin{proof}
First to prove that $\vl$ is an extension of $VL$.  All that
is required is to prove the $\sob$-COMP and $\sob$-rec
axioms.  To prove $\sob$-COMP, note that every $\sob$
formula $\phi$ is equivalent to an open formula $\phi'$.
For example, 
$$VL \proves \exists z < b \psi(z,\vec x, \vec X) \iff \psi(
f_{\psi,b}(\vec x, \vec X), \vec x, \vec X )$$
when
$\psi$ is an open formula.  Then the function $F_{\phi',t}$
is the witness for $$\exists Z \le t \forall i<t [ Z(i) \iff
\phi(i) ].$$  To prove the $\sob$-rec axiom, we can define a
function $f(i,a,E)$ that returns the $i$th node in the path
the axiom says exists.
The function $f$ can be defined using $p$-bounded number
recursion.  From there, a function witnessing the $\sob$-rec
axiom can be defined.

To prove that the extension is conservative,
we show how to take any model $M$ of $VL$ and
find an expansion that is a model of $\vl$.  The idea is
to expand the model one function at a time.
We can order the functions in $\lfl$ such that each function
is defined in terms of the previous functions.  Let
$\mathcal L_i$ be the language $\la$ plus the first $i$
functions in $\lfl$. Let $M_i$ be the model obtained by
expanding $M$ to the functions in $\mathcal
L_i$. 
We will show that the model $M_\infty = \bigcup M_i$ is a model $\vl$.
A similar proof can be found in Chapter 9 of \cite{CN06} and we will not
repeat it here.

\end{proof}

\subsection{Quantified Propositional Calculus}

We are also interested in quantified propositional proof
systems.  The proof systems we use were originally defined
in \cite{KP90}, and then they were redefined in \cite{CM04,
Morioka05}, which is the presentation we follow.

The set of connectives are $\set{\AND,
\OR, \neg, \exists, \forall, \true, \false }$, where $\true$
and $\false$ are constants for true and false, respectively.
Formulas are built using these connectives in the usual way.
We will often refer to formulas by the number of quantifier
alternations.
\begin{defn}
The set of formulas $\Sigma_0^q = \Pi_0^q$ is the set of
quantifier-free propositional formulas.  For $i>0$, the
set of $\Sigma_i^q$ ($\Pi_i^q$) formulas is the smallest set of formulas
that contains $\Pi_{i-1}^q$ ($\Sigma_{i-1}^q$) and is closed
under $\AND$, $\OR$,
existential (universal) quantification, and if $A\in
\Pi_i^q$ ($A\in \Sigma_i^q$) then $\neg A \in \Sigma_i^q$
($\neg A \in \Pi_i^q$).
\end{defn}

The first proof system, from which all others will be
defined, is the proof system $G$.  This proof
system is a sequent calculus based on Gentzen's system $LK$.
The system $G$ is essentially the DAG-like, propositional
version of $LK$.  We will not give all of the rules, but
will mention a few of special interest.

The cut rule is
\begin{prooftree}
\AX$ A, \Gamma \fCenter \Delta$
\AX$ \Gamma \fCenter \Delta, A$
\LL{cut}
\BI$\Gamma \fCenter \Delta$
\end{prooftree}
In this rule, we call $A$ the cut formula.
There are also four rules that introduce quantifiers:
\begin{prooftree}
\AX$A(x), \Gamma \fCenter \Delta$
\LL{$\exists$-left}
\UI$\exists z A(z), \Gamma \fCenter \Delta$

\AX$ \Gamma \fCenter \Delta,A(B)$
\LL{$\exists$-right}
\UI$ \Gamma \fCenter \Delta,\exists z A(z)$
\noLine
\BIC{}
\end{prooftree}

\begin{prooftree}
\AX$ \Gamma \fCenter \Delta,A(x)$
\LL{$\forall$-left}
\UI$ \Gamma \fCenter \Delta, \forall z A(z)$

\AX$A(B), \Gamma \fCenter \Delta$
\LL{$\forall$-right}
\UI$\forall z A(z), \Gamma \fCenter \Delta$
\noLine
\BIC{}
\end{prooftree}
These rules have conditions on them.  In $\exists$-left and
$\forall$-right, the variable $x$ must not appear in the
bottom sequent.  In these rules, $x$ is called the
eigenvariable.  In the other two rules, the formula $B$ must
be a $\Sigma_0^q$ formula, and no variable that appears
free in $B$ can be bound in $A(x)$.

The initial sequents of $G$ are sequents of the form
$\fCenter \true$, $\false \fCenter$, or $x
\fCenter x$, where $x$ is any propositional variable.  A $G$ proof is a
series of sequents such that each sequent is either an
initial sequent or can be derived from previous sequents
using one of the rules of inference.  The proof system $G_i$ is $G$ with cut formulas
restricted to $\Sigma_i^q$ formulas.

We define $G^*$ as the treelike version of $G$.  So,
a $G^*$ proof is a $G$ proof where each sequent in used as
an upper sequent in an inference at most once.  A $G_i^*$ proof is
a $G^*$ proof in which cut formulas are prenex $\Sigma_i^q$.
In \cite{Morioka05}, it was shown that, for treelike proofs,
it did not matter if the cut formulas in $G_i^*$ were prenex
or not.  So when we construct $G_i^*$ proofs, the cut formulas
will not always be prenex, but that does not matter.

To make proofs simpler, we assume that all treelike proofs are in {\em
free-variable normal form}.
\begin{defn}
\label{def:fvnf}
A parameter variable for a $G_i^*$ proof $\pi$ is a variable that appears
free in the final sequent of $\pi$.
A proof $\pi$ is in {\em free-variable normal form} if (1) every non-parameter
variable is used as an eigenvariable exactly
once in $\pi$, and (2) parameter variables are not used as eigenvariables.
\end{defn}
Note that, if a proof is treelike, we can always put it in
free-variable normal form by simply renaming variables.  In
fact, $VPV$ proves that every treelike proof can be put in
free-variable normal form.

A useful property of these proof systems is the {\em subformula
property}.  It can be shown in $VL$ that every formula in a
$G_i^*$ proof is an ancestor (and therefore a subformula) of a cut formula or a formula
in the final sequent.  This is useful because it tells us
that any non-$\Sigma_i^q$ formula in a $G_i^*$ proof must be
an ancestor of a final formula.

\subsection{Truth Definitions}

In order to reason about the proof systems in the theories,
we must be able to reason about quantified propositional
formulas.  We follow the presentation in
\cite{Krajicek95, KP90, CM04}.

%If $F$ is a formula, we will use $\code F$ as the string encoding of the formula. 
Formally formulas will be coded as strings, but we will not
distinguish between a formula and its encoding.  So if $F$
is a formula, we will use $F$ as the string encoding the
formula as well.
The method of coding a formula can be found in \cite{CM04}.  
%If $A$ is an assignment, then $\code A$ will be the string encoding of an assignment.  

In this paper, we are only interested in $\Sigma_0^q$
formulas and prenex $\Sigma_1^q$ formulas.
For $\Sigma_0^q$ formulas, we are able to give an $\sob(\lfl)$ functions
that evaluates the formula.  This formula will be referred to
using $A \satisfies_0 F$, where $A$ is an assignment and $F$
is a formula.  We leave the precise definition to the
readers.

Given a prenex $\Sigma_1^q$ formula $F$, the truth
definition is a formula that says there is an assignment to
the quantified variables that satisfies the $\Sigma_0^q$
part of the formula. This formula will be referred to as  $A
\satisfies_1 F$.

Valid formulas (or tautologies) are defined as
\begin{align*}
TAUT_i(F)  \equiv \forall A, (\text{``A is an assignment to the variables of $F$''}  \implies A \satisfies_i F)
\end{align*}

This truth definition can be extended to define the truth of
a sequent.  So, if $\Gamma \fCenter \Delta$ is a sequent of
$\Sigma_i^q \cup \Pi_i^q$ formulas, then
\begin{equation*}
\begin{split}
(A \satisfies_i ~ \Gamma \fCenter \Delta) \equiv &
  \text{``there exists a formula in $\Gamma$ that $A$ does not satisfy'' } \\
  & \OR \text{ ``there exists a formula in $\Delta$ that $A$ satisfies''}
\end{split}
\end{equation*}

Another important formula we will use is the reflection
principle for a proof system.  We define the $\Sigma_i^q$
reflection principle for a proof system $P$ as
\begin{equation*}
\begin{split}
\Sigma_i^q \text{-RFN}(P) \equiv \forall F \forall \pi, ( \text{``$\pi$ is a $P$ proof of $F$''} \AND F \in \Sigma_i^q ) \implies TAUT_i(F)
\end{split}
\end{equation*}
This formula essentially says that, if there exists a $P$ proof of a
$\Sigma_i^q$ formula $F$, then $F$ is valid.  Another way of putting it
is to say that $P$ is sound when proving $\Sigma_i^q$ formulas.

\subsection{Propositional Translations}

There is a close connection between the theory $V^i$ and the
proof system $G_i^*$.  You can think of $G_i^*$ as the
non-uniform version of $V^i$.  This idea might not make much
sense at first until you realize you can translate a $V^i$
proof into a polynomial-size family of $G_i^*$ proofs.
The translation
that we use is described in \cite{Cook05,CM04}.  It is
a modification of the Paris-Wilkie translation \cite{PW85}.
Given a $\Sigma_i^B$ formula $\phi(\vec x, \vec
X)$ over the language $\la$, we want to
translate it into a family of propositional formulas
$||\phi(\vec
x, \vec X)||[\vec m; \vec n]$, where the size of the formulas is
bounded by a polynomial in $\vec m$ and $\vec n$.
The formula $||\phi(\vec x, \vec X)||[\vec m; \vec n]$ is
meant to be a formula that is a tautology when $\phi(\vec x, \vec X)$ is true in the standard model whenever
$x_i = m_i$ and $|X_i| = n_i$.  Then if $\phi(\vec x, \vec X)$ is true in the standard model for all $\vec x$ and $\vec X$,
then every $||\phi(\vec x, \vec X)||[\vec m; \vec n]$ is a tautology.  

The variables $\vec m$ and $\vec n$ will often be omitted since
they are understood.  The free variables in the propositional formula will
be $p_j^{X_i}$ for $j < n_i-1$.  The variable $p_j^{X_i}$ is
meant to represent the value of the $j$th bit of $X_i$; we
know that
the $n_i$th bit is 1, and for $j>n_i$, we know the $j$th bit is
0.
The definition of the translation proceeds by structural induction on $\phi$.

Suppose $\phi$ is an atomic formula.  Then it has one of the
following
forms:  $s = t$, $s < t$, $X_i(t)$, or one of the
trivial formulas $\false$ and $\true$, for terms
$s$ and $t$.  
Note that
the terms $s$ and $t$ can be evaluated immediately.
This is because the exact value of every number variable and the size of
each string variable is known.  Let $val(t)$ be value
of the term $t$.

In the first case, we define $||s = t||$ as the formula $\true$, if $val(s) = val(t)$, and
$\false$, otherwise.
A similar construction is done for $s<t$.  If $\phi$ is one of
the trivial formulas, then $||\phi||$ is the same
trivial formula.  So now, if $\phi \same X_i(t)$,
let $j= val(t)$.  Then the translation is defined as follows:
$$||\phi|| \same \begin{cases}
p_j^{X_i} & \text{if } j < n_i - 1 \\
1 & \text{if } j = n_1 - 1 \\
0 & \text{if } j > n_1 - 1 \\
\end{cases}$$

Now for the inductive part of the definition.  Suppose $\phi
\same \alpha \AND \beta$.  Then $$||\phi|| \same
||\alpha|| \AND ||\beta||.$$
When the connective is $\OR$ or $\neg$, the definition is similar.
If the outermost connective is a number
quantifier bound by a term $t$, let $j=val(t)$.  Then the translation is
defined as
\begin{equation*}
\begin{split}
||\exists y \le t, \alpha(y)|| & \same
\bigvee
\limits_{i=0}^j ||\alpha(y)||[i] \\
||\forall y \le t, \alpha(y)|| &
\same \bigwedge
\limits_{i=0}^j ||\alpha(y)||[i] \\
||\exists Y \le t, \alpha(Y)|| & \same
\exists p_0^Y,
\dots, \exists p_{m-2}^Y, \bigvee
\limits_{i=0}^j ||\alpha(Y)||[i] \\
||\forall Y \le t, \alpha(Y)|| & \same
\forall p_0^Y,
\dots, \forall p_{m-2}^Y, \bigwedge
\limits_{i=0}^j ||\alpha(Y)||[i]
\end{split}
\end{equation*}

Now we are able to state the translation theorem for $V^i$ and $G_i^*$.
\begin{thm}
Suppose $V^i \proves \phi(\vec x, \vec X)$, where $\phi$ is
a bounded formula.  Then there are
polynomial-size $G_i^*$ proofs of the family of tautologies
$$||\phi(\vec x, \vec X)||[\vec m; \vec n].$$
\end{thm}
This type of theorem is the standard way of proving that
the reasoning power of the proof system is as least as
strong as that of the theory.

%%%%%%%%%%%%%%%%%%%%%%%%%%%%%%%%%%%%%%%%%%%%%%%%%%%%%%%%%%%
%%%%%%%%%%%%%%%%%%%%%%%%%%%%%%%%%%%%%%%%%%%%%%%%%%%%%%%%%%%
%%%%%%%%%%%%%%%%%%%%%%%%%%%%%%%%%%%%%%%%%%%%%%%%%%%%%%%%%%%
%%%%%%%%%%%%%%%%%%%%%%%%%%%%%%%%%%%%%%%%%%%%%%%%%%%%%%%%%%%

\section{Definition of $GL^*$}

In this section, we will define the proof system we
wish to explore.  As was stated in the introduction, this
proof system is defined by restricting cut formulas to
a set of formulas that can be evaluated in $L$.  Alone that
is not enough to change the
strength of the proof system, so we also restrict the
use of eigenvariables.

The first step is to define a set of formulas that can be
evaluated in $L$.  These formula will be bases on $CNF(2)$
formulas.  A $CNF(2)$ formula is a $CNF$ formula where no
variable has more than two occurrences in the entire formula.
It was shown in \cite{Johannsen04} that determining whether or not a
given $CNF(2)$ formula is satisfiable is complete for $L$.
Based on this we get the following definition:
\begin{defn}
\label{sqcnf}
The set of formulas $\sqcnf$ is the smallest set
\begin{enumerate}
\item containing $\Sigma_0^q$,
\item containing every formula $\exists \vec z, \phi( \vec
z,
\vec x )$
where (1) $\phi$ is a quantifier-free CNF
formula $\bigwedge_{i=1}^m C_i$ and (2) existence
of a
$z$-literal $l$ in $C_i$ and $C_j$, $i \neq j$,  implies
existence of an $x$-variable $x$ such that $x \in C_i$ and
$\neg
x \in C_j$ or vice versa, and
\item closed under substitution of $\Sigma_0^q$ formulas
that contain only $x$-variables
for $x$-variables.
\end{enumerate}
\end{defn}
\begin{defn}
The idea behind this definition is that any assignment to
the variables $\vec x$ reduces the quantifer-free protion to
a $CNF(2)$ formula in $\vec z$.
$GL^*$ is the propositional proof system $G^*_1$ with cuts
restricted to
$\sqcnf$ formulas in which every free variable in a
non-$\Sigma_0^q$ cut formula is a parameter variable.
\end{defn}

The restriction on the free variables in the cut formula
might seem strange, but it is necessary.  If we did not
have this restriction, then the proof system would be as
strong as $G_1^*$.  We will not give a full proof of this,
but the interested reader can see information on $GPV^*$ in
\cite{Perron07}.  What we will show is that, if the
restriction on the variables is not present, then the proof
system can simulate $G_1^*$ for $\Sigma_1^q$ formulas. 

Let $H^*$ be the proof system $G^*_1$ with cuts restricted to
$\sqcnf$ formulas and no restriction on the free variables.
\begin{defn}
An extension cedent $\Lambda$ is a sequence of formulas
\begin{equation}
\label{lam}
\Lambda \same y_1 \iff B_1, y_2 \iff B_2,
\dots, y_n \iff
B_n\end{equation}
where $B_i$ is a $\Sigma_0^q$ formula that does not mention any of the
variables $y_i,\dots,y_n$.
We call the variables $y_1,\dots,y_n$ extension variables.
\end{defn}

Based on a lemma in \cite{Krajicek95}, Cook and Nguyen
proved the following lemma in \cite{CN06}.
\begin{lem}
\label{EPK_G1}
If $\pi$ is a $G_1^*$ proof of $\exists \vec z A(\vec z, \vec
x)$, where $A$ is a $\Sigma_0^q$ formula,
then there exists a $PK$ proof $\pi'$ of $$\Lambda \fCenter A(\vec
y, \vec x)$$ where $\Lambda$ is as in \ref{lam} and $|\pi'|
\le p(|\pi|)$, for some polynomial $p$.
\end{lem}

The proof guaranteed by this lemma is also an $H^*$ proof
since every $PK$ proof is also an $H^*$ proof.
Extending this proof with a number of applications of $\exists$-right, we
get an $H^*$ proof of
\begin{equation}
\label{h_1}
\Lambda \fCenter \exists \vec z~ A(\vec z, \vec x).
\end{equation}

So now we need to find a way to remove the extension cedent
$\Lambda$.  This is done one formula at a time.  Suppose $y
\iff B$ is the last formula in $\Lambda$.  The key
observation is that $\exists y [ y \iff B ]$ is a $\sqcnf$
formula because the formula can be express as $\exists y [ (y
\OR \neq B) \AND (\neg y \OR B)]$.  So we can apply $\exists$-left with $y$ as the
eigenvariable to \eqref{h_1}.  The eigenvariable restriction
is met because $y$ is the last eigenvariable, and, therefore,
cannot appear anywhere else the extension cedent.  Then we cut $\exists y [ y
\iff B]$ after deriving $\fCenter \exists y [ y \iff B]$.
We can then do this for every formula is $\Lambda$ starting
at the end.
This proves the following theorem.
\begin{thm}
$H^*$ $p$-simulates $G^*_1$ for $\Sigma_1^q$ formulas.
\end{thm}
This proof is not always a $GL^*$ proof because the
extension variables are not parameter variables, yet they
appear in cut formulas.

%\subsection{Definition of $GNL^*$}
%
%The set of formulas that are complete for $NL$ are the
%$\Sigma Krom$ formulas.  These formula are based on $2-CNF$
%formulas.  A $2-CNF$ formula is a $CNF$ formula where no
%clause has more than two variables.
%It is well known that determining whether or not a
%given $2-CNF$ formula is satisfiable is complete for $NL$.
%Based on this we get the following definition:
%\begin{defn}
%\label{krom}
%The set of formulas $\Sigma Krom$ is the smallest set
%\begin{enumerate}
%\item containing $\Sigma_0^q$,
%\item containing every formula $\exists \vec z, \phi( \vec
%z,
%\vec x )$
%where $\phi$ is a quantifier-free CNF
%formula in which no clause contains more that two
%$z$-literals, and is
%\item closed under substitution of $\Sigma_0^q$ formulas
%that contain only $x$-variables
%for $x$-variables.
%\end{enumerate}
%\end{defn}
%Then $GNL^*$ is defined as follows:
%\begin{defn}
%$GNL^*$ is the propositional proof system $G^*$ with cuts
%restricted to
%$\Sigma Krom$ formulas in which every free variable in a
%non-$\Sigma_0^q$ formula is a parameter variable.
%\end{defn}

%%%%%%%%%%%%%%%%%%%%%%%%%%%%%%%%%%%%%%%%%%%%%%%%%%%%%%%%%%%
%%%%%%%%%%%%%%%%%%%%%%%%%%%%%%%%%%%%%%%%%%%%%%%%%%%%%%%%%%%
%%%%%%%%%%%%%%%%%%%%%%%%%%%%%%%%%%%%%%%%%%%%%%%%%%%%%%%%%%%
%%%%%%%%%%%%%%%%%%%%%%%%%%%%%%%%%%%%%%%%%%%%%%%%%%%%%%%%%%%

\section{Adjusting $VL$}

In order to prove the translation theorem, we start with the theory
$VL$, which
corresponds to $L$ reasoning.  This theory was defined in
Section \ref{sec:def}.  The proof of the translation theorem
is similar to other proofs of its type.  We take an anchored
(or free-cut free) proof.  Then the cut formulas in
this proof will translate into the cut formulas in the
propositional proof.  If we use $VL$ for this, there are two
problem: (1) not all of the axioms of $VL$ translate into
$\sqcnf$ formulas and (2) the restriction of the free
variables in cut formulas
may not be met.  In the first subsection, we take care of the first
problem.  The second problem in taken care of in Section
\ref{sec:prob2}.

\subsection{A New Axiomatization For $VL$}

We want to reformulate the axioms of $VL$ so
they translate
into $\sqcnf$ formulas.
All of the 2BASIC axioms are
$\sob$, so they translate into
$\Sigma_0^q$ formulas, which are $\sqcnf$, so they do not
create any problems.
We only need to consider
$\sob$-COMP and $\Sigma_0^B$-rec.  We handle
$\Sigma_0^B$-COMP the same way Cook and
Morioka
did in \cite{CM04}.  That is, if the proof system is asked
to
cut the translation of an instance of the $\sob$-COMP axiom,
then the propositional proof is changed so that the cut
becomes
$\bigwedge_{i=0}^t [||\phi(i)|| \iff ||\phi(i)||]$, which is
$\sqcnf$.  To take care of $\sob$-rec, we define a new
theory that is equivalent to $VL$ by replacing the
$\sob$-rec axiom.

Informally the new
axiom says that there exists a string $Z$ that gives a
specific pseudo-path of
length $b$
in the graph with $a$ nodes and edge relation $\phi(i,j)$.
This path starts at node $0$.  If $(i,j)$ is an edge in this
path, then $j$ is the smallest number with an edge from
$i$ to $j$, or $j=a$ when there are no outgoing edges.  Note
that
the edge may not exist in the original graph when $j=a$.
This
is why we call it a pseudo-path.
If $(i,j)$ is the $w$th edge in the path, then
$Z(w,i,j)$ is true, and $Z(w,i',j')$ is false for every
other
pair.  
This is described by the $\sob$-edge-rec axiom
scheme:
\begin{equation}
\tag{$\sob$-edge-rec}
\label{erec}
\exists Z \le 1+\pair{b,a,a}[\rho_1 \AND \rho_2 \AND
\rho_3
\AND \rho_4 \AND \rho_5 \AND \rho_6 \AND \rho_7 \AND
\rho_8],
\end{equation}
where
\begin{equation*}
\begin{split}
\rho_1 \same & \forall j < a, \neg Z(0,0,j)
\OR \phi(0,j) \OR \exists l < j \phi(0,l) ) \\
\rho_2 \same & \forall j\le a \forall k < j, \neg Z(0,0,j)
\OR \neg \phi(0,k) \OR \exists l < k \phi(0,l) ) \\
\rho_3 \same & \forall i \le a \forall j \le a, i=0 \OR \neg
Z(0,i,j)
 \\
\rho_4 \same & \forall w < b \forall i \le a \forall j \le
a, \neg Z(w+1,i,j)
\\ & \OR \exists h \le a Z(w,h,i) \OR \neg \phi(i,j) \OR
\exists
l <
j \phi(i,l) \\
\rho_5 \same & \forall w < b \forall i\le a \forall j < a,
\neg
Z(w+1,i,j)
\OR \phi(i,j) \OR \exists l < j \phi(i,l) \\
\rho_6 \same & \forall w < b \forall i \le a \forall j \le a
\forall k
< j,
\neg Z(w+1,i,j)
\OR \neg \phi(i,k) \OR \exists l < k \phi(i,l) \\
\rho_7 \same & \exists i \le a \exists j \le a, Z(b,i,j) \\
\rho_8 \same & \forall \pair{w,i,j} \le \pair{b,a,a}, [ w >
b \OR i
> a \OR j > a ] \implies \neg Z(w,i,j)
\end{split}
\end{equation*}
and $\phi(i,j)$ is a $\sob$ formula that does not mention
$Z$, but may have other free variables.  It is not
immediately obvious that the axiom says what it is suppose to, so
we will take a closer look.

Let $Z$ be a string that witnesses the axiom.  We want to
make sure $Z$ is the path described above.  Looking at
$\rho_3$, we see the path starts at $0$.  Suppose $Z(0,0,j)$
is true.  We must show that $j$ is the first node adjacent
to $0$.  This follows from $\rho_1$, which guarantees
$\phi(i,j)$ is true when $j<a$, and $\rho_2$, which
guarantees
$\phi(i,k)$ is false when $k<j$.  A similar argument can be
made with $\rho_5$ and $\rho_6$ to show that every node is
the smallest node adjacent to its predecessor.  To make
sure the path is long enough, we have $\rho_7$, which says
there is a $b$th edge, and $\rho_4$, which says if there is
a $(w+1)$th edge there is a $w$th.  As you may have noticed,
there are parts of this formula that semantically are not
needed.  For example, the $\exists l < j \phi(0,l)$ in
$\rho_1$ is not needed.  It is used to make sure
the axiom translates into a $\sqcnf$ formula.  We add
$\rho_8$ to make sure there is a unique $Z$ that witnesses
this axiom.

\begin{notation}
For simplicity, $\psi_\phi$ is the
$\sob$ part of the $\sob$-edge-rec axiom instantiated with
$\phi$. Note this includes the bound on the size of $Z$.  So
the axiom
can be
written as $\exists Z \psi_\phi$.
\end{notation}
\begin{defn}
$VL'$ is the theory axiomatized by the axioms of $V^0$, the
$\sob$-edge-rec axioms, and Axiom (\ref{const}).  The
language of
$VL'$ is the language of $V^0$ plus a string
constant $C$ with defining axiom
\begin{equation}
\label{const}
\size C = 0
\end{equation}
\end{defn}
We add the string constant to the language so we can put
$VL'$ proofs
in
free variable normal form (below).  We do not use the
constant for any other reason.  Also, in the translation,
we can treat $C$ as a string variable with $n=0$.

\begin{lem}
\label{equiv}
The theory $VL$ is equivalent to $VL'$.
\end{lem}
\begin{proof}
To prove the two theories are equivalent, we must show that
$VL$
proves the
$\sob$-edge-rec axiom and that $VL'$ proves the $\sob$-rec
axiom.  Since the two axioms express similar ideas, this is
not surprising.

To show that $VL$ proves the $\sob$-edge-rec axiom, let
$\phi(i,j)$ be
any $\sob$ formula.  Then let $Y$ be the string such
that
$Y(i,j)
\iff (j<a \implies \phi(i,j)) \AND \forall k<j \neg
\phi(i,k)$.
This $Y$ exists by $\sob$-COMP.
We can think of $Y$ as the graph that contains only the
edges the
$\sob$-edge-rec axiom would use.  Since $VL$ proves the
$X-MIN$
formula, it follows that $VL$ proves $\forall i \le a,
\exists j
\le a, Y(i,j)$.  This means there exists a path of
length
$b$
in $Y$ that starts at node $0$ Lemma \ref{lem:start}.  It is
a simple task to verify the $b$ edges in
this path satisfy the $\sob$-edge-rec axiom for $\phi$.

To show that $VL'$ proves the $\sob$-rec axiom, let
$\phi(i,j)$ be a $\sob$ formula such that 
$\forall i \le
a \exists j\le a, \phi(i,j)$.  By the $\sob$-edge-rec axiom,
there
is a pseudo-path of length $b$ in the graph $\phi$.  We need
to
show
that this is a real path.  Suppose $(i,j)$ is an edge in
the
path.  If $j<a$, then $(i,j)$ is in the graph by $\rho_1$
and
$\rho_5$.  Otherwise, $j=a$, and $\forall k<j \neg
\phi(i,k)$.
This implies $\phi(i,j)$ since every node has out-degree at
least 1.  This means every edge in
the
pseudo-path exists, and there exists a path of length $b$.
\end{proof}

The next step is to be sure the translation of the
$\Sigma_0^B$-edge-rec axiom is
a
$\sqcnf$ formula.  This is done by a careful inspection of
the formula.
\begin{lem}
The formula $||\exists Z \psi_\phi(a,b,Z)||$ is a $\sqcnf$
formula.
\end{lem}
\begin{proof}
First we assume $\phi(i,j) \same X(i,j)$ for some variable
$X$.
It is easy to see that $||\psi_{X(i,j)}(a,b,Z)||[a,b; t, a*a]$,
where $t$ is the bound on $Z$ given in the $\sob$-edge-rec
axiom, is a CNF
formula.  Note that we assigned $|Z|=t$ and $|X|=a*a$.  We
now need to make sure the clauses have the correct
form.  This is done by examining each occurrence of a
bound
literal.  To verify this, the proof will require a careful
inspection
of the definition of the axiom.
The only bound variables are those that come from $Z$.
These are $p^Z_{w,i,j}$, which we will refer to as
$z_{w,i,j}$.  The only free variables are those
corresponding to $X$.  These variables will be referred to as
$x_{i,j}$.  We will first look at the positive occurrences of
$z_{w,i,j}$.
On inspection, we can observe that, when $w<b$, every
occurrence of
$z_{w,i,j}$ must be in clauses that are part of the
translation
of $\rho_4$.  We want to show that every clause that is
part of the
translation of $\rho_4$ has conflicting free variables.
This
is true since $\neg X(i,j_1)$ will conflict with one of
the
variables from $\exists l<j_2, X(i,l)$ when $j_1 <
j_2$.
When
$w=b$, the variable $z_{b,i,j}$ appears once in $\rho_7$.
Now we turn to the negative occurrences.  When $w=0$, the
variable $z_{0,i,j}$ will
appear negatively in the clauses corresponding to
$\rho_1$,
$\rho_2$, and $\rho_3$.  If $i > 0$, it will appear only
in
the clauses corresponding to
$\rho_3$ and will appear only once.  If $i=0$, the
variable
$z_{0,0,j}$
will not appear in the translation of $\rho_3$ because the
$i=0$
part will satisfy the clause.  It is easy to observe that
every occurrence of the variable in the translation of
$\rho_1$
and $\rho_2$ will
have a conflicting free variable.  Examine the
construction $X(0,j) \OR \exists l < j X(0,l)$ at
the
end
of $\rho_1$ and $\neg X(0,k) \OR \exists l<k X(0,l)$
at the
end of $\rho_2$.  A similar argument can be
made with $\rho_4$, $\rho_5$, and $\rho_6$ when $w > 0$.
This implies that the translation is a $\sqcnf$ formula
when
$\phi(i,j) \same X(i,j)$.  When $\phi$ is a more general
formula, the translation is the formula in the first case
with the free variables substituted with the translation
of
$\phi$, which will be $\Sigma_0^q$.  Since $\sqcnf$
formulas
are
closed under this type of substitution, the formula is
$\sqcnf$ in
all cases.
\end{proof}

\subsection{Normal Form For $VL'$}
\label{sec:prob2}

In this section, we want to find a normal form for $VL'$
proofs that makes sure the translation of $VL'$ proofs
satisfy the variable restriction for $GL^*$.
The normal form we want is {\em cut variable normal form}
(CVNF)
and is
defined in the following.

\begin{defn}
A formula $\phi(Y)$ is bit-dependent on $Y$ if there is an
atomic sub-formula of $\phi$ of the form $Y(t)$, for some
term $t$.
\end{defn}

\begin{defn}
\label{fvnf}
A proof is in free variable normal form if (1) every
non-parameter free variable $y$ or $Y$ that appears in the
proof is
used as an eigenvariable exactly once and (2) parameter
variables are never used as eigenvariables.
\end{defn}
Note that if a proof is in free variable normal form we can
assume that every instance of the non-parameter variable $Y$
(or $y$) is in an ancestor of the sequent where $Y$ is used
as an eigenvariable.  If it is not, we can replace $Y$ with
the constant $C$ in all those sequents.

\begin{defn}
A cut in a proof is anchored if the cut formula is an
instance of an axiom.
\end{defn}

\begin{defn}
\label{cvnf}
A $VL'$ proof $\pi$ is in \em{cut variable normal form} if
$\pi$
is (1) in free variable normal form, (2) every cut with a
non-$\sob$ cut formula is anchored, and (3) no
cut formula that is an instance of the $\sob$-edge-rec axiom
is bit-dependent on a non-parameter free string variable.
\end{defn}
It is known how to find a proof with the first two
properties \cite{CN06,Buss98}, and this part will not be
repeated here.  Instead we focus on how to find a proof
satisfying the third property.
\begin{thm}
\label{norm_form}
For every $\Sigma_1^B$ theorem of $VL'$ there exists a $VL'$-proof of that formula in CVNF.
\end{thm}

The proof of this theorem is the most technical in this
paper.  At a high level, it amounts to showing
$\sob$-edge-rec is closed under substitution of strings
defined by $\sob$-edge-rec and $\sob$-COMP. We
begin with an anchored proof that is in free variable
normal form.  We want to change every cut that violates
condition (3) in the definition of CVNF.  Consider the
proof given in Figure \ref{bad_proof}.
\begin{figure}
\begin{prooftree}
\AxiomC{$P$}
\noLine
\UIC{$\ddots\vdots\iddots$}
\noLine
\UI$\exists Z \psi_{\phi(Y)}(Z), \gamma(Y), \Gamma \fCenter
\Delta$
\AXC{$\ddots\vdots\iddots$}
\noLine
\UI$\gamma(Y), \Gamma \fCenter \Delta, \exists Z
\psi_{\phi(Y)}(Z)$
\BI$\gamma(Y), \Gamma \fCenter \Delta$
\UI$\exists Y \gamma(Y), \Gamma \fCenter \Delta$
\end{prooftree}
\caption{Example of a proof that is not in CVNF}
\label{bad_proof}
\end{figure}
This is a simple example of what can go wrong.  The general
case is handled in the same way, so we will only consider
this case.

Since all $\Sigma_1^B$ cut formulas are anchored and the
$\exists Y \gamma(Y)$ must eventually be cut, it is be an
instance of $\sob$-COMP or $\sob$-edge-rec.  So you can
think of $\gamma$ as a formula that completely defines $Y$.
Then we want to change $\phi(Y)$ so that it does not mention
$Y$ explicitly, but instead uses the definition of $Y$ given
by $\gamma$.  Note that, for this to be true, the final formula must be $\Sigma_1^B$; otherwise, $Y$ could have been used as an eigenvariable in a $\forall$-right inference and would not be well defined.
\begin{lem}
\label{main}
For any $\sob$ formula $\phi(Y)$, there exist $\sob$
formulas $\phi_1$ and $\phi_2$ such that $\phi_1$ is not
bit-dependent on $Y$ and $V^0$ proves the sequent
$$\gamma(Y), \psi_{\phi_1}(Z), \forall i < t [ Z'(i) \iff
\phi_2(Z) ] \fCenter \psi_{\phi(Y)}(Z').$$
\end{lem}
\begin{proof}
This proof is divided into two cases.  In the first case, we
assume
\begin{equation}
\label{star}
\gamma(Y) \same \size Y \le t \AND \forall i < t [Y(i) \iff
\phi'(i)].
\end{equation}
That is, $\exists Y \gamma(Y)$ is an instance of $\sob$-COMP.  We
know $Y$ must appear in that position because it eventually gets
quantified.  In this case, $\phi_1$ is $\phi$ with every
atomic formula of the form $Y(s)$ replaced by $s < t \AND
\phi'(s)$, and $\phi_2$ is the formula $Z(i)$.  We can prove that there
exists a $V^0$ proof of \eqref{star} by structural induction
on $\phi$. 

For the second case, we assume $\gamma(Y) \same
\psi_{\phi'}(Y).$  That is, $Y$ is the pseudo-path in the
graph of $\phi'$.  The first step is to define branching
programs that compute $Y$ and $Z'$ (the pseudo-path in the graph of
$\phi$) using $Y$.  Then $\phi_1$ is the $\Sigma_0^B$
description of the composition of these branching programs,
and $\phi_2$ is the $\Sigma_0^B$ formula that extracts $Z'$
from the run of this last branching program.

\begin{defn}
A branching program is a nonempty set of nodes labeled with
triples $(\alpha, i,j)$, where $\alpha$ is a $\sob$ formula
over
some set of variables and $0 \le i,j \le t$ for some term
$t$
that depends only on the inputs to the program.  Semantically, if
a node $u$ is labeled with $(\alpha,
i,j)$, then, when the branching program is at node $u$, it will
go to node $i$, if $\alpha$ is true, or node $j$,
otherwise. The
initial node is always $0$.
\end{defn}
Note that a branching program is essentially a graph with a
special form, and, as with graphs, we use families of
branching programs that can be described by $\Sigma_0^B$
formulas.  However, we will not give the explicit
construction of the formula; we leave it to the reader.

The first step is to introduce the initial branching program
$BP_0$ that computes $Z'$.  The nodes of $BP_0$ are interpreted as triples
$\pair{w,i,j}$.  An invariant for this branching program is
that, if we reach the node $\pair{w,i,j}$, then the $w$th
node of $Z'$ is $i$ and $\forall k <j \neg \phi(i,k)$.
At each node, we check if $j$ is the next node. Let $a$
be the maximum value of a node and $b$ be the length of the
path.  This means the number of nodes in $BP_0$ is bound
by
$\pair{b,a,a}$.  So now to define the labels.  If $j <
a$, then $\pair{w,i,j}$ is labeled with $( \phi(i,j),
\pair{w+1,j,0}, \pair{w,i,j+1})$.  If $j=a$, then
$\pair{w,i,j}$ is labeled with $( \true, \pair{w+1,j,0}, 0 )$.  It is
easy to see that the invariants hold and that $Z'$ can be
obtained from a path in $BP_0$ using $\Sigma_0^b$-COMP.

The branching program that computes $Y$ is constructed the
same way except $\phi'$ is used instead of $\phi$.  Let this
branching program be $BP$.

Moving on to the second step, we now want to simplify $BP_0$ so that every node whose label is
bit-dependent on $Y$
is labeled with an
atomic formula.  This is done to simplify the construction
of the composition.  We start with $BP_0$.  Then, given $BP_i$, we
define $BP_{i+1}$ by removing one connective in a
node of $BP_i$ that is not in the right form.  Let node $n$
in $BP_i$ be labeled with
$(\alpha, u_1,u_2)$.  The construction is divided into five cases: one
for each possible outer connective.

{\em Case $\alpha \same \neg \beta$:} $BP_{i+1}$ is the same as
$BP_i$ except node $n$ is now labeled with $(\beta,u_2,u_1)$.

{\em Case $\alpha \same \beta_1 \AND \beta_2$:} 
The nodes of $BP_{i+1}$ are interpreted as
pairs $\pair{u,v}$.  The node $\pair{u,0}$ 
corresponds to node $u$ in $BP_i$. 
The label of $\pair{n,0}$ becomes $(\beta_1,
\pair{n,1}, \pair{u_2,0})$ and the label for $\pair{n,1}$ is
$(\beta_2, \pair{u_1,0}, \pair{u_2,0})$.  Notice that
$\pair{n,1}$ is used as an intermediate node while evaluating
$\alpha$.

{\em Case $\alpha \same \beta_1 \OR \beta_2$:}
$BP_{i+1}$ is defined as in the previous case, with a few minor modifications.  This
case is left to the reader.

{\em Case $\alpha \same \exists z \le t \beta(z)$:}
The nodes become pairs as in the previous case, but this time the
labels are different.
The node $\pair{n,i}$ is labeled with
$(\beta(i), \pair{u_1,0}, \pair{n,i+1})$, when $i < t$.
If $i=t$, the node is labeled with $(\beta(i), \pair{u_1,0},
\pair{u_2,0})$.  In this case, the branching program is
looking for an $i$ that satisfies $\beta(i)$.

{\em Case $\alpha \same \forall z \le t \beta(z)$:}
This case is similar to the previous case.  The only
difference is the branching program is looking for an $i$
that falsifies $\beta(i)$.

Let $BP_n$ be the final branching program in this
construction above.  We now construct a branching
program $BP'$ that is the composition of $BP_n$ and $BP$.
The nodes of $BP'$ are pairs $\pair{u_1,u_2}$ where the
first element corresponds to a node in $BP_n$ and the second
element corresponds to a node in $BP$.  

Suppose node $u_1$ in
$BP_n$ is labeled with $(\alpha, v_1, v_2)$.  If $\alpha$ is
not bit-dependent on $Y$, then the node $\pair{u_1,0}$ is
labeled with $(\alpha, \pair{v_1,0}, \pair{v_2,0})$.

It is also possible that $\alpha$ is bit-dependent on $Y$;
in which case, $\alpha$ is of the form $Y(w,i,j)$.  Let
$(\beta, w_1, w_2)$ be the label for node $u_2$ in $BP$.
Then the node $\pair{u_1, u_2}$ is labeled as follows:

\begin{tabbing}
$(\beta, \pair{u_1,w_1}, \pair{u_1,w_2})$, \= if $u_2 \le
\pair{w,a,a}$ and $u_2 \neq \pair{w,i,j}$, \\
$(\beta, \pair{v_1,0}, \pair{v_2,0})$ \> if $u_2 =
\pair{w,i,j}$, and \\
$(\true, \pair{v_2,0}, \pair{v_2,0})$  \> otherwise.
\end{tabbing}

In this case, we are using the second element to run $BP$
and determine if the $w$th edge in the path is $(i,j)$.  If
it is, we move on to $\pair{v_1,0}$, and, if it is not, we
move on to $\pair{v_2,0}$.  In the labels above, the first
line corresponds to running $BP$.  The second line
corresponds to a check if $(i,j)$ is the $w$th edge.  The
third line is used when we have already found the $w$th edge
and it is not $(i,j)$.

It is not difficult to see that it is possible to construct
$\phi_1$ (a
$\Sigma_0^B$ formula describing $BP'$), and $\phi_2$ (a
formula extracting $Z'$ from a run of $BP'$.  Moreover,
$V^0$ proves that this construction works.
\end{proof}
%%%%%%%%%%%%%%%%%%%%%%%%%%%%%%%%%%%%%%%%%%%%%%%%%%%%%%%%%%%%
%%%%%%%%%%%%%%%%%%%%%%%%%%%%%%%%%%%%%%%%%%%%%%%%%%%%%%%%%%%%
Using this lemma, we are able to change the proof in Figure
\ref{bad_proof} into the proof in Figure \ref{new_proof}.
In that proof, $P'$ is the proof $P$ with the rules that
introduced $\exists Z$ ignored (renaming variables if
necessary), and $Q$ is an anchored $V^0$
proof, which we know exists by the lemma above.
This gives us a new proof of the same formula that still
satisfies properties (1) and (2) in Definition \ref{cvnf}
and it contains one less cut that is bit-dependent on $Y$.
\begin{figure}
\begin{prooftree}
\AxiomC{$P'$}
\noLine
\UIC{$\ddots\vdots\iddots$}
\noLine
\UI$\psi_{\phi(Y)}(Z), \gamma(Y), \Gamma \fCenter
\Delta$
\AxiomC{$Q$}
\noLine
\UIC{$\ddots\vdots\iddots$}
\noLine
\UI$\gamma(Y),  \psi_{\phi_1}(Z), \tau(Z')
\fCenter \Delta,
\psi_{\phi(Y)}(Z)$
\doubleLine
\BI$\psi_{\phi_1}(Z), \tau(Z'),
\gamma(Y), \Gamma \fCenter \Delta$
\UI$\psi_{\phi_1}(Z), \exists Z' \tau(Z'),
\gamma(Y), \Gamma \fCenter \Delta$
\end{prooftree}

\begin{prooftree}
\AX$\psi_{\phi_1}(Z), \exists Z' \tau(Z'),
 \gamma(Y), \Gamma \fCenter \Delta$
\AX$\fCenter
\exists Z'\tau(Z') $
\doubleLine
\BI$\psi_{\phi_1}(Z), \gamma(Y), \Gamma \fCenter \Delta$
\UI$\exists Z \psi_{\phi_1}(Z), \gamma(Y), \Gamma \fCenter
\Delta$
\end{prooftree}

\begin{prooftree}
\AX$\exists Z \psi_{\phi_1}(Z), \gamma(Y), \Gamma \fCenter
\Delta$
\AX$\fCenter \exists Z \psi_{\phi_1}(Z)$
\doubleLine
\BI$\gamma(Y), \Gamma \fCenter
\Delta$
\UI$\exists Y \gamma(Y), \Gamma \fCenter \Delta$
\end{prooftree}

\caption{Modification of the proof in Figure
\ref{bad_proof}.  The formula $\tau(Z')$ is used to replace $\forall i
< t [ Z'(i) \iff \phi_2(Z) ]$}
\label{new_proof}

\end{figure}

%%%%%%%%%%%%%%%%%%%%%%%%%%%%%%%%%%%%%%%%%
%% Proof of main theorem
%%%%%%%%%%%%%%%%%%%%%%%%%%%%%%%%%%%%%%%%%%
Using this manipulation, we prove Theorem \ref{norm_form}.

\begin{proof}[Proof of Theorem \ref{norm_form}]
It would be nice to be able to simply say we can repeatedly
apply
the manipulations above and eventually the proof will be in
CVNF, but this is not obvious.  In the
manipulation, if $\gamma(Y)$ is bit-dependent on a string
variable other than $Y$,
then the new $\sob$-edge-rec cut formula is bit-dependent on
that variable.  This
includes non-parameter string variables.  So we need to
state our
induction hypothesis more carefully.

Let $Y_1,\dots,Y_n$
be all of the non-parameter free string variables that appear
in $\pi$ ordered such that the variable $Y_i$ is
used
as a eigenvariable
before $Y_j$ for $i<j$.  This implies $Y_i$ does not appear
in
$\gamma(Y_j)$
in the manipulations above.  So now suppose no
$\sob$-edge-rec cut formula is bit-dependent on the
variables $Y_1,\dots,Y_k$, for
some $k<n$.  Then
we can manipulate $\pi$ such that the same holds for the
variables
$Y_1,\dots,Y_{k+1}$. To accomplish this, we simply
manipulate
every $\sob$-edge-rec cut formula that is bit-dependent on
$Y_{k+1}$ as
described above.  Since $Y_1,\dots,Y_k$ cannot appear in
$\gamma(Y_{k+1})$, those variables will not violate the
condition.
So by induction, we can get a proof that is in CVNF.
\end{proof}

\section{Translation Theorem}

We are now prepared to prove the translation theorem.  The
proof is done by induction on
the
length of the proof.  For the base case, we need to prove
the
translation of the axioms of $VL'$.  We know the $\sob$-COMP
and the 2BASIC axioms have polynomial-size $G_0^*$ proofs
from other translation theorems \cite{CM04}.  This means
they also have polynomial-size $GL^*$ proofs.
Axiom (\ref{const}) is easy to prove since it
translates to $\fCenter \true$.  We still need to show how
to prove the
$\sob$-edge-rec
axiom in $GL^*$.  Recall that we write the axiom as $\exists
Z \psi_\phi(a,b,Z)$.  Note that the axiom does have a bound 
on $Z$, but it has been omitted since the specific bound is
not important.

\begin{lem}
\label{gl_edge_rec}
The formula $||\exists Z \psi_\phi(a,b,Z)||$ has a $GL^*$
proof
of size $p(a,b)$ for some polynomial $p$.
\end{lem}
\begin{proof}
The proof is done by a brute force induction.  We prove, in
$GL^*$, that, if there exists a pseudo-path of length $b$,
then
there
exists a pseudo-path of length $b+1$.  It is easy to prove
there exists a pseudo-path of length $0$.  Then with repeated
cutting we get our final result.  The entire path is
quantified, so we do not cut formulas with non-parameter free
variables.

Given variables that encode a
path of length $b$, we can define $\Sigma_0^q$ formulas that
determine the next edge.
Let $A_{i,j} \same
||\phi(i,j)||$.
Since $\phi$ is a $\sob$ formula, $A_{i,j}$ is a
$\Sigma_0^q$
formula.  To prove that there is an edge that starts the path,
consider the formula
$$B_{0,0,j} \same A_{0,j} \AND \bigwedge
\limits_{k=0}^{j-1} \neg
A_{0,k},$$
when $j<a$, and
$$B_{0,0,a} \same \bigwedge \limits_{k=0}^{a-1} \neg
A_{0,k}.$$  It is
easy to see $B_{0,0,j}$ is true for exactly one $j\le a$.
This is
also provable in $GL^*$.  This shows that $GL^*$ has a
polynomial-size proof of
$$||\exists Z \psi_\phi(a,1,Z)||.$$

For the inductive step, if there is a path of length $b$
and
the
path is given by the variables $z_{w,i,j}$, then the
witnesses for the next edge
are defined as follows:
$$B_{b+1,i,j}
\same \bigvee_{k=0}^a z_{b,k,i} \AND A_{i,j} \AND
\bigwedge_{k=0}^{j-1} \neg A_{i,k},$$
when $j<a$, and
$$B_{b+1,i,a}
\same \bigvee_{k=0}^a z_{b,k,i} \AND
\bigwedge_{k=0}^{a-1} \neg A_{i,k}.$$
Using the fact that exactly one $z_{b,i,j}$ is true, we
can
prove
in $GL^*$ that exactly one $B_{b+1,i,j}$ is true.
This shows that $GL^*$ has a polynomial-size proof of
$$||\exists Z \psi_\phi(a,b,Z)|| \fCenter ||\exists Z
\psi_\phi(a,b+1,Z)||.$$

So now we are able to prove $||\exists Z
\psi_\phi(a,b,Z)||$
for
any $b$ by successive cutting.  Recall that $||\exists Z
\psi_\phi(a,b,Z)||$ is a $\sqcnf$ formula, and note that
the
free
variables in $||\exists Z \psi_\phi(a,b,Z)||$ do not
change
as
$b$ changes.  This means we are allowed to do the cut.
\end{proof}

This can be used to prove the translation theorem.
\begin{thm}[$VL$-$GL^*$ Translation Theorem]
Suppose $VL$ proves $\exists Z < t \phi(\vec x, \vec X,Z)$,
where $\phi$ is a $\Sigma_0^B$ formula.
Then there are polynomial-size $GL^*$ proofs of $||\exists
Z<t
\phi(\vec x, \vec X,Z)||[\vec n]$.
\end{thm}
\begin{proof}
By Theorem \ref{equiv} and Theorem \ref{norm_form}, there exists
a
$VL'$ proof $\pi$ of $\exists Z < t \phi(\vec x, \vec X,Z)$
that is in
CVNF.

We proceed by induction on the depth of $\pi$.  The base
case
follows from Lemma \ref{gl_edge_rec} and the comments that
precede it.  The inductive step is divided into cases: one
for
each rule.  With the exception of cut, every rule can be
handled
the same way it is handled in the $V^1$-$G_1^*$ Translation
Theorem
(Theorem 7.51, \cite{CN06}), and will not be repeated here.

When looking at the cut rule, there are three cases.  If the
cut
formula
is $\sob$, then we simply cut the corresponding $\Sigma_0^q$
formula in
the $GL^*$ proof.  If the cut formula is not $\sob$, then it
must
be anchored since the proof is in CVNF.
This
means the cut formula is an instance of $\sob$-edge-rec
or an instance of $\sob$-COMP.  First suppose it is an
instance
of $\sob$-edge-rec.  Then we are able to cut the
corresponding
formula in the $GL^*$ proof. This is because the axiom
translates into a
$\sqcnf$ formula, and the free variables in the
translation are parameter variables since the formula is not
bit-dependent on
non-parameter string variables.

When the cut formula is an instance of $\sob$-COMP, we apply
the same transformation as in the proof of the
$VNC^1$-$G_0^*$ translation theorem \cite{CM04}. That
is,
we remove the quantifiers by replacing the variables with
$\Sigma_0^q$ formulas that witness the quantifiers.  This
change does not effect other cuts since their free
variables are parameter variables or they are $\Sigma_0^q$
formulas and remain $\Sigma_0^q$ after the
substitution.
The current cut
formula becomes a $\Sigma_0^q$ formula, which can be cut.
Note that, since there are a constant number of cuts of this
axiom, the substitution does not cause an exponential
increase in the size of the formulas.
\end{proof}

\section{Proving Reflection Principles}

In this section, we show that $GL^*$ does not capture
reasoning for a higher complexity class.  This is done by
proving, in $VL$, that $GL^*$ is sound.  This idea comes
from \cite{Cook75}, where Cook showed that $PV$ proves
extended-Frege is sound, and \cite{KP90}, where Krajicek and
Pudlak showed $T^i_2$ proves $G_i$ is sound for $i>0$.

We will actually show that $\vl$ proves $GL^*$ is
sound.  The idea behind the proof is to give an
$\lfl$ function that witnesses the quantifiers in the proof.
Then we prove, by $\Sigma_0^B(\lfl)$-IND, that this
functions witness every sequent, including the final
sequent.  Therefore the formula is true.

We start by giving an algorithm that
witnesses $\sqcnf$ formulas in $L$ when the formula is true.
This algorithm is the algorithm given in \cite{Johannsen04} with a
few additions to find the satisfying assignment.
We describe an $\lfl$ function that corresponds to this
algorithm and prove it correct in
$\vl$.  We then
use this function to find an $\lfl$ function that
witnesses $GL^*$ proofs, and prove it correct in $\vl$.

\subsection{Witnessing $\sqcnf$ Formulas}

Let $\exists \vec z A(\vec x, \vec z)$ be a $\sqcnf$
formula.  We will describe how to find a witness for this
formula.  
We assume that $A$ is a $CNF$ formula.  That
is, the substitution of the $\Sigma_0^q$ formulas has not
happened.  The general case is essentially the same.

The first thing to take care of is the encoding of $A$.
We will not go through this is detail.  Suffice it
to say that parsing a formula can be done in $TC^0$
\cite{CM04}, and, as long as we are working in a theory that extends
$TC^0$ reasoning, we can use any reasonable encoding.
We will refer to the $i$th
clause of $A$ as $C^A_i$.  A clause will be viewed as a set
of literals.  A {\em literal} is either a variable or its
negation.  So we will write $l \in C^A_i$ to mean that the
literal $l$ is in the $i$th clause of $A$.  Since the
parsing can be done in $TC^0$, these formulas can be defined
by $\Sigma_0^B(\lfl)$ formulas.  An assignment will also be
viewed as a set of literal.  If a literal is in the set,
then that literal is true.  So an assignment $X$ satisfies a
clause $C$ if and only in $X \intersect C \neq \emptyset$.

Given values for $\vec x$, we first simplify $A$ to get a
$CNF(2)$ formula.  We will refer to the simplified formula
as $F$.  This can
be done using the $\lfl$ function defined by the following
formula:
$$l \in C_i^F \iff l \in C_i^A \AND X \intersect C_i^A =
\emptyset,$$
where $X$ is the assignment to the free variables.
From the definition of a $\sqcnf$ formula, $\vl$ can easily
prove that $F$ now encodes a $CNF(2)$ formula.  In fact, it
can be shown that no literal appears more than once.  A
satisfying assignment to this formula is the witness we
want.  Mark Bravermen gave an algorithm for finding this assignment \cite{Bravermen03}, but we use a different algorithm that is easier to formalize.

Before we describe the algorithm that finds this assignment,
we go through a couple definitions.  First, a {\em pure
literal} is a literal that appears in the formula, but its
negation does not.  Next the formula imposes an order on
the literals.  We say a literal $l_1$ {\em follows } a
literal $l_2$ if the clause that contains $l_1$
also contains $l_2$, and $l_1$ is immediately to the right
of $l_2$, circling to the beginning if $l_2$ is the last
literal.  More
formally:
\begin{equation*}
\begin{split}
follows( l_1, l_2, F ) \iff \exists i,  l_1 \in C_i^F \AND l_2 \in C_i^F & \AND \forall l_3 ( l_2 < l_3 < l_1 \implies l_3 \not \in C_i^F ) \\
                                                                       & \AND \forall l_3 ( l_3 < l_1 < l_2 \implies l_3 \not \in C_i^F ) \\
                                                                       & \AND \forall l_3 ( l_1 < l_2 < l_3 \implies l_3 \not \in C_i^F ) \\
\end{split}
\end{equation*}
Note
that if a clause contains a single literal then that literal
follows itself.  
Also, note that literals are coded by numbers and $l_1 < l_2$
means the number coding $l_1$ is less then the number coding
$l_2$.

To find the assignment to $F$, we will go through the
literals in the formula in a very specific order.  Starting
with a literal $l$ that is not a pure literal, the {\em next literal} is the
literal that follows $\overline l$:
$$next(l_1,F) = l_2 \iff follows( l_2, \overline l_1,F).$$
Note that if $l_1$ is a pure literal, then there is no next
literal, so we simply define it to be itself.
The important distinction is that $next$ gives an ordering
of the literals in a formula, and $follows$ orders the
literal in a clause.  When $F$ is understood, we will not
mention $F$ in $next$ and $follows$.

The algorithm that finds the assignment works in stages.  At
the beginning of
stage $i$, we have an assignment that satisfies the first
$i-1$ clauses.  Then, in the $i$th stage, we make local
changes to this
assignment to satisfy the $i$th clause as well.
At a high level, to satisfy the $i$th clause, we start with
the first literal in the $i$th clause, and assign that literal to
true.  The clause that contains this literal's negation may
be have gone from being satisfied to being unsatisfied.  So
we now go to the next literal, which is in this other clause.
We continue this until we get to a point where we know the
other clause is satisfied.  We need to be able to do this in
$L$.  Algorithm \ref{algo:cnf2} shows how to do this.
\begin{algorithm}
\caption{Algorithm for Stage $i$}
\label{algo:cnf2}
\begin{algorithmic}
\STATE Set $l_1$ to the first literal in clause $i$.
\REPEAT
\STATE Assign true to $l_1$.
\STATE set $l_2 := next(l_1)$
\WHILE{ $l_2$ is not the complement of $l_1$ }
\STATE Assign true to $l_2$
\STATE set $l_2 := next(l_2)$
\STATE If $l_2$ is a pure literal, assign true to $l_2$, and stage $i$ is done.
\STATE If $l_1$ and $l_2$ are in the same clause, stage $i$ is done.
\ENDWHILE
\STATE Assign true to $l_1$. \COMMENT{This statement is redundant, but it is included to emphasis that $l_1$ is true.}
\STATE set $l_1 := next(l_1)$
\UNTIL{$l_1$ is the first literal in clause $i$}
\STATE At this point we know the formula is unsatisfiable.
\end{algorithmic}
\end{algorithm}
At any point in the algorithm, the only information we need
are the values of $l_1$ and $l_2$, so this is in $L$.  Note
that we do not store the assignment on the work tape, but on
a write-only, output tape.
What is not obvious is why this algorithm
works.

The next lemma can be used to show that the both loops will
eventually finish.
\begin{lem}
For all literals $l$, there exists a $t>0$ such that after $t$
applications of $next$ to $l$, we get to $l$ or a pure literal.
\end{lem}
\begin{proof}
Let $next^0(l) = l$ and $next^{t+1}(l) = next(next^{t}(l))$.
Since $next$ has a finite range, there exist a minimum $i$
and $t$ such that $next^i(l) = next^{i+t}(l)$.  Suppose this is not a pure literal.  If $i > 0$,
then $next(next^{i-1}(l)) = next( next^{i+t-1}(l))$.
However, this implies $next^{i-1}(l) = next^{i+t-1}(l)$
since $next$ is one-to-one when not dealing with pure literals.  This violates our choice of
$i$. Therefore $i=0$, and $l= next^0(l) = next^t(l)$.
\end{proof}
The implies the inner loop will halt, because, if it does
not end earlier, $l_2$ will eventually equal $l_1$ which
both will be in the same clause.  For the outer loop, if the
algorithm does not halt for any other reason, $l_1$ will
eventually return to the first literal in the $i$th clause.

The next lemma plays a small role in the proof of
correctness.
\begin{lem}
\label{lem:something}
Suppose the algorithm fails at stage $i$ and that
$next^t(l') = l$, where $l'$ is the first literal in clause
$i$.  Then, for every literal in the same clause as $l$,
there is a $t'$ such that $next^{t'}(l')$ equals that
literal.
\end{lem}
\begin{proof}
To prove this lemma, we will show that there exists a $t'$
that equals the literal that follows $l$.  Then by
continually applying this argument, you get that every
literal in the clause is visited.

Let $l'$ be the first literal in the $i$th clause.  
Then, after going through the outer loop $t$ times, $l_1 =
l$.  Since the algorithm fails, the inner loop will finish
because $l_2 = \overline l_1$.  This means there is a $t'$
such that $next^{t'}(l') = \overline l$.  Then
$next^{t'+1}(l')$ is the literal that follows $l$.
\end{proof}

\begin{thm}
\label{thm:correctness1}
If the algorithm fails, the formula is unsatisfiable.
\end{thm}
\begin{proof}
This is proved by contradiction.  Let $F$ be a $CNF(2)$ formula
and $A$ be an assignment that satisfies it.
Assume that the algorithm fails.  From this we can defined a
function from the set of variables to the set of clauses as
follows:
$$f(i) = j \iff (x_i \in C_j^F \AND x_i \in A) \OR (\neg x_i
\in C_j^F \AND \neg x_j \in A ).$$
Informally, if $f(i)=j$ then clause $C_j$ is true because of
the variable $x_i$.  Since the formula is satisfied, this
function is onto the set of clauses.  Also, since $F$ is
$CNF(2)$, no literal appear more than once.  So $f$ is
indeed a function because if $f(i) = j$ and $f(i) = j'$ then
the literal $x_i$ or $\neg x_i$ is in both $C_j^F$ and
$C_{j'}^F$.

Now we will use the assumption that the algorithm fails to
find a way to restrict $f$ so that it violates the $PHP$.
Suppose the algorithm fails at stage $i$.  Let $l$ be first
literal in clause $i$.  We then define sets of variables
$V^a$ as follows:
$$V^a = \set{ x_n : \exists b < a~next^{b}(l) = x_n \OR
next^{b}(l) = \neg x_n }.$$
We also defined sets of clauses $W^t$ as follows:
$$W^a = \set{ C_n : \exists x \in V^a ( x \in C_n \OR \neg x \in
C_n) }.$$
Note that for a large enough $a$, say $\size{F}$, if $C_n$
is in $W^a$, then every variable that appears in $C_n$ is in
$W^a$ by \ref{lem:something}.
We show by induction on $a$ that $\size {V^a} < \size {W^a}$.

For $a = 1$, $\size {V^a} = 1$.  If $l$ is a pure literal or
$l$ and $\neg l$ are in the same clause,
then the algorithm would succeed.  Otherwise $\size{W^a} =
2$.

For the inductive case, suppose $\size {V^a} < \size {W^a}$.
Let $l' = next^{a+1}(l)$.  If $l'$ is not a new variable,
then $\size{V^{a+1}} = \size {V^a} < \size {W^a} =
\size{W^{a+1}}$.  If $l'$ is a new variable, then $\overline l'$
must be in a new clause.  For, if this was not the case, the
algorithm would succeed.  
To see this, let $l_1$ be the most recent literal in the same clause as $l'$.  We know $l_1$ is not $\overline l'$ since $l'$ is a new variable. 
Then eventually $l_2$ will become $next(l')$, which is in the same clause as $l_1$.
The inner loop will not end because $l_2$ becomes the complement of $l_1$ since that would mean $next(\overline l_1)$ is more recent.

This gives $\size{V^{a+1}} = \size
{V^a} + 1 < \size {W^a} + 1 = \size{W^{a+1}}$.

If we restrict $f$ to $V^{\size F}$, then $f$ is a function
from $V^{\size {F}}$ that is onto $W^{\size {F}}$ 
violating the $PHP$.
\end{proof}

\begin{thm}
\label{thm:correctness2}
If the algorithm succeeds, then, for all
$i$, the assignment after given at the end of stage $i$
satisfies the first $i$ clauses of $F$.
\end{thm}
\begin{proof}
The proof is done by induction on $i$.  For $i=0$, the
statement holds since there are no clauses to satisfy.
As an induction hypothesis, suppose
the statement holds for $i$.  Then we will show if
the algorithm ever visits one of the literals in clause $n$, then that
clause is satisfied.  

Consider clause $n$, where $n\le i+1$.  Find the last point
in the algorithm that either $l_1$ or $l_2$ was in clause
$n$, and let $l$ be that literal.
First, it is possible that when the algorithm ends $l_2$ is in clause $n$.  If $l_2$ is a pure literal, then $l_2$ is set to true, satisfying the clause.  Otherwise, $l_1$ and $l_2$ are in the same clause.  In this case, $l_1$ is true since it was assigned true.  If $l_2$ ever became $\overline l_1$, the algorithm would exit the inner loop, so $\overline l_1$ could never have been assigned true.

Second, we consider the possibility that $l_2$ was not in clause $n$ when the algorithm ended.
Then we claim that
$l$ is true, and, therefore, clause $n$ is satisfied.
Suppose for a contradiction that it
is not.  Then at some later point $\overline l$ was assigned
true.  This could happen in one of three places.  First is
if $l_1 = l$ and we are at the beginning of the outer loop.
However, $l_2$ would be set to $next(\overline l)$ right
after, which is in clause $n$. This means we did not find
the last occurrence of a literal in clause $n$ as we should
have.  A similar argument can be used in the other two
places.
\end{proof}

We now turn to formalizing this algorithm.  For this, we
define an $\lfl$ function $f(i,t)$ that will return the
value of $l_1$ and $l_2$ after $t$ steps in stage $i$.  This
is done using number recursion.  In the following let
$f(c,t) = \pair{l_3,l_4}$:
\begin{equation*}
\begin{split}
f(i,0) = \pair{ l_1, l_2 } \iff & l_1 = \min \limits_l l \in
C_i^F \AND l_2 = next(l_1) \\
f(c,t+1) = \pair{ l_1, l_2 } \iff
     & \phi_1 \implies l_1 = next(l_3) \AND l_2 = next(l_1) \\
\AND &  \neg \phi_1 \AND \phi_2 \implies (l_1 = l_3 \AND l_2 = l_4) \\
\AND &  \neg \phi_1 \AND \neg \phi_2 \implies (l_1 = l_3 \AND l_2 = next(l_4) )
\end{split}
\end{equation*}
where
\begin{equation*}
\begin{split}
\phi_1 & \equiv l_3 = \overline l_4 \\
\phi_2 & \equiv (sameClause(l_3,l_4) \OR pureLiteral(l_4) ) \\
\end{split}
\end{equation*}
The formulas $\phi_1$ and $\phi_2$ are the conditions that
are used to recognize when the inner loop ends.  The first
formula is when the loop ends and we have to continue with
the outer loop.  The second formula is when the stage is
finished.  In the formula version, we do not stop if the
algorithm fails.  Instead we view the algorithm as failing
if after $\size F ^2$ steps, $\phi_2$ was never true.  We
use this value since $\size F$ is an upper bound on the
number of literals in $F$ and current state of the algorithm
is determined by a pair of literal.
In the following, any reference to time has the implicit
bound of $\size F^ 2$.

The final step is to extract the assignment.  
The assignment is done by finding the last time a variable
is assigned a value.  This means we must be able to determine when a variable is
assigned a value.  To do this, observe that a literal is
assigned true just before the $next$ function is applied
to that literal.  With this is mind we get the following:
\begin{equation*}
\begin{split}
Assigned(i, t, l) \iff \exists l', & f(i,t) = \pair{next(l), l'} \OR f(i,t) = \pair{l', next(l)}
\end{split}
\end{equation*}
So $Assigned(i, t, l)$ means that $l$ was assigned true
during the $t$th step of stage $i$.
Then we can get the assignment as follows:
\begin{equation*}
\begin{split}
l \in Assignment( i, F ) \iff & c = \max \limits_c \exists t~Assigned(c, t, l) \\
                          \AND & t = \max \limits_t Assigned(c, t, l) \\
                          \AND & c' = \max \limits_{c'} \exists t'~Assigned(c, t', \overline l) \\
                          \AND & t' = \max \limits_{t'} Assigned(c', t', \overline l) \\
                          \AND & (c > c' \OR (c=c' \AND t > t') )
\end{split}
\end{equation*}
The idea is the value of a variable is the last value that
was assigned to it.

The $\vl$ proof that this algorithm is correct is the
essentially the same as the proofs of Theorem
\ref{thm:correctness1} and Theorem \ref{thm:correctness2}, which can be formalized in $\vl$.  This
gives the following.

\begin{thm}
$\vl$ proves that, if the algorithm fails, the formula is unsatisfiable.
\end{thm}

\begin{thm}
$\vl$ proves that, if the algorithm succeeds, then, for all
$i$, $Assignment(i, F)$ gives a
satisfying assignment to the first $i$ clauses of $F$.
\end{thm}

\subsection{Witnessing $GL^*$ Proofs}

Let $\pi$ be a $GL^*$ proof of a $\Sigma_1^q$ formula
$\exists \vec z P(\vec x, \vec z)$, and let $A$ be an
assignment to the parameter variables.
We assume $\pi$ is in free variable normal form (Definition \ref{def:fvnf}).

% To prove the soundness of $GL^*$, a function that
% witnesses every sequent in $\pi$ will be defined. Then
% we can prove by induction that every sequent is indeed
% witnessed.

Let $\Gamma_i \fCenter \Delta_i$ be the $i$th sequent in
$\pi$.  We will prove by induction that for any assignment
to all of the free variables of $\Gamma_i$ and $\Delta_i$,
a function $Wit(i,\pi,A)$ will find at least one formula that satisfies
the sequent.

% \begin{defn}
% A sequent $\Gamma \fCenter \Delta$ is
% \textit{witnessed} by a function $Wit$, if there exists a
% $\Sigma_1^q$ formula in $\Delta$ that is witnessed
% by $Wit$, whenever every formula in $\Gamma$ is true,
% every
% $\Sigma_0^q$ formula in $\Delta$ is false, and $Wit$ is
% given an assignment to all of the free variables in the
% sequent.
% \end{defn}

There are two things to note.  By the subformula property, every formula in $\Gamma_i$
is $\sqcnf$, which means it can be evaluated.  Also, we need
an
assignment that gives appropriate values to the
non-parameter free variables that could appear.  To take
care of this second point, we extend $A$ to an assignment
$A'$ as follows:
\begin{algorithmic}[1]
\STATE Given a non-parameter free variable $y$,
find the $\exists$-left inference in $\pi$ that uses $y$
as an eigenvariable.  Let $z$ be the new bound variable
and let
$F$ be the principal formula.
\STATE Find the descendant of $F$ that is used as a cut
formula.  Let $F'$ be the cut formula.  Note
that $F$ is a subformula of $F'$, and, because of the
variable restriction on cut formulas,
every free variable in $F'$ is a parameter
variable.
\STATE Assign $y$ the value that $Assignment(F',A)$
assigns $z$.
\end{algorithmic}
The reason for this particular assignment will become
evident in the
proof of Lemma \ref{thm_wit}.

We can now define $Wit(i,\pi,A')$, which witnesses $\Gamma_i
\fCenter \Delta_i$.  $Wit$ will go through each formula in
the sequent to find a formula that satisfies the sequent.
$\sqcnf$ formulas are evaluated using the algorithm
described in the previous section.  We will now focus
our attention on other $\Sigma_1^q$ formulas, which must
appear in $\Delta_i$.  Each
$\Sigma_1^q$ formula $F \same \exists \vec z F^*(\vec z)$ in
$\Delta$ is evaluated by finding a witness to the
quantifiers as follows:
\begin{algorithmic}[1]
% \IF{ $F$ is $\sqcnf$ }
% \STATE the witness for $F$ is $W(F,A')$
% \ELSE
\STATE Find a formula $F'$ in $\pi$ that is an ancestor
of $F$, is satisfied by $A'$, and is a $\Sigma_0^q$ formula
of the form $F^*( z_1/B_1, \dots, z_n/B_n ),$ where each
$B_i$ is $\Sigma_0^q$
\STATE $z_i$ is assigned $\true$ if $A'$ satisfies $B_i$,
otherwise it is assigned $\false$
\STATE if no such $F'$ exists, then every bound variable is
assigned $\false$.
% \ENDIF
\end{algorithmic}

\begin{lem}
\label{thm_wit}
For every sequent $\Gamma_i \fCenter \Delta_i$ in $\pi$,
$Wit(i,\pi, A')$ finds a false formula in $\Gamma_i$ or a
witness for a formula in $\Delta_i$.
\end{lem}
\begin{proof}
We prove the theorem by induction on the depth of the
sequent.  For the base case, the sequent is an axiom, and
the theorem obviously holds.  For the inductive step, we
need to look at each rule.  We can ignore $\forall$-left
and $\forall$-right since universal quantifiers do not
appear in $\pi$.

We will now assume all formulas in $\Gamma_i$ are true and
all $\sqcnf$ formulas in $\Delta_i$ as false.  So we need
to find a $\Sigma_1^q$ formula in $\Delta_i$ that is true.

Consider cut.  Suppose the inference is
\begin{prooftree}
\AX $F,\Gamma \fCenter \Delta$
\AX $\Gamma \fCenter \Delta, F$
\BI $\Gamma \fCenter \Delta$
\end{prooftree}
First suppose $F$ is true.  By induction, with the upper
left
sequent, $Wit$ witnesses one of the formulas in $\Delta$.
Then the corresponding formula in the bottom sequent is
witnessed by $Wit$.  This is because the ancestor of the
formula in the upper sequent that gives the witness is also
an ancestor of the corresponding formula in the lower
sequent.  If $F$ is false, it cannot be the formula that
was witnessed in the
upper
right sequent, and a similar argument can be made.

Consider $\exists$-right.  Suppose the inference is
\begin{prooftree}
\AX $\Gamma \fCenter \Delta, F(B)$
\UI $\Gamma \fCenter \Delta, \exists zF(z)$
\end{prooftree}

First suppose $F(B)$ is $\Sigma_0^q$.  If it is false, we
can apply the inductive hypothesis, and, by an argument
similar to the previous case, prove one of the formulas in
$\Delta$ must be witnessed.  If $F(B)$ is true, then $Wit$
will witness $\exists z F(z)$ since $F(B)$ is the
ancestor that gives the witness.  If $F(B)$ is not
$\Sigma_0^q$, then we can apply the inductive hypothesis,
and, by the same argument, find a formula that is witnessed.

The last rule we will look at is $\exists$-left. Suppose the
inference is
\begin{prooftree}
\AX $F(y), \Gamma \fCenter \Delta$
\UI $ \exists zF(z), \Gamma \fCenter \Delta$
\end{prooftree}
To be able to apply the inductive hypothesis, we need to be
sure that $F(y)$ is satisfied.  If $\exists z F(z)$ it true,
then we know $F(y)$ is satisfied by the construction of
$A'$: the value assigned to $y$ is chosen
to satisfy $F(y)$ if it is possible.  Otherwise, $\exists
z F(z)$ is false, and we do not need induction.

For the other rules the inductive hypothesis can be applied
directly and the witness found as in the previous cases.
\end{proof}

\begin{thm}
$\vl$ proves $GL^*$ is sound for proofs of $\Sigma_1^q$
formulas.
\end{thm}
\begin{proof}
The functions $Assignment$ and $Wit$ are in $FL$ and can
be formalized in $\vl$.  A function that finds $A'$, given
$A$, can also be formalized since it in $\vl$.  The
final thing to note is that the proof of Lemma
\ref{thm_wit} can be formalized in $\vl$ since the
induction hypothesis can be express as a $\sob(\lfl)$
formula and the induction carried out.
\end{proof}

The reason this proof does not work for a larger
proof system, say $G_1^*$, is because $Assignment$ cannot be
formalized for the larger class of cut formulas.  Also, if
the variable restriction was not present, we would not be
able to find $A'$ in $L$, and the proof would, once
again, break down.

% ----------------------------------------------------------------
\bibliography{gl}
\bibliographystyle{plain}
% ----------------------------------------------------------------
\end{document}